\documentclass[12pt,runningheads]{llncs}

\usepackage{fullpage}
\usepackage[T1]{fontenc}
\usepackage{amssymb,bm,color,xspace,mathtools}
\usepackage[colorlinks=true]{hyperref}
\usepackage[dvipdfmx]{graphicx}
\usepackage{algorithm}
\usepackage{algorithmicx}
\usepackage{braket}
\usepackage{comment}

\usepackage{comment}
\usepackage{bm}
\usepackage[mathlines]{lineno}
\let\doendproof\endproof
\renewcommand\endproof{~\hfill$\qed$\doendproof}

\newcommand{\rmO}{\mathrm{O}}
\newcommand{\relmiddle}[1]{\mathrel{}\middle#1\mathrel{}}
\newcommand{\Bz}{\mathrm{Bz}}
\newcommand{\SSS}{\mathrm{SS}}

\begin{document}
  \title{Computing Power Indices in Weighted Majority Games with Formal Power Series}
\titlerunning{Computing Power Indices in Weighted Majority Games}
\author{Naonori Kakimura\inst{1}\orcidID{0000-0002-3918-3479} \and
Yoshihiko Terai\inst{1}}
\authorrunning{N. Kakimura and Y. Terai}
\institute{Keio University, \\Yokohama 223-8522, Japan.\\
\email{kakimura@math.keio.ac.jp}, 
\email{11ty24@keio.jp}
%\url{http://www.springer.com/gp/computer-science/lncs}
}
\maketitle              % typeset the header of the contribution
\begin{abstract}
In this paper, we propose fast pseudo-polynomial-time algorithms for computing power indices in weighted majority games.
We show that we can compute the Banzhaf index for all players in $\rmO(n+q\log (q))$ time, where $n$ is the number of players and $q$ is a given quota.
Moreover, we prove that the Shapley--Shubik index for all players can be computed in $\rmO(nq\log (q))$ time. 
Our algorithms are faster than existing algorithms when $q=2^{o(n)}$.
Our algorithms exploit efficient computation techniques for formal power series.

\end{abstract}

\section{Introduction}\label{sec:intro}

In this paper, we study measuring the power of players in weighted majority games.
The weighted majority game~(also known as the weighted voting game) is a mathematical model of voting in which each player has a certain number of votes.
In the game, each player votes for or against a decision, and the decision is accepted if the sum of players' voting is greater or equal to a fixed quota.
%The game is a special case of a cooperative game, called ``simple games'' by von Neumann and Morgenstern~\cite{}.

In 1964, Shapley and Shubik~\cite{ShapleyShubik1954} proposed how to measure a voting power of each player,
called the \textit{Shapley--Shubik index}.
The Shapley--Shubik index is a specialization of the Shapley value~\cite{shapley:book1952} for cooperative games.
Another concept for measuring voting power was introduced by Banzhaf~\cite{Banzhaf1965}, called the \textit{Banzhaf index}.
These power indices are used to evaluate political systems such as the European Constitution and the IMF~\cite{ALGABA20071752,KOCZY2012152,kurz2016computing}. 
%There are many other measurements to assess the voting power of each agents such as ...

It is known that computing the Banzhaf and Shapley--Shubik indices are both known to be NP-hard~\cite{DengPapadimitriou94,PrasadKelly}.
This paper is thus concerned with calculating these power indices in pseudo polynomial time.
%A naive implementation of calculating the Banzhaf index is to enumerate all the coalitions.
%This requires $O(n 2^n)$ time for all players.
%Similarly, computing the Shapley-Shubik indices for $n$ players requires $O(n\cdot n!)$ time.
To devise fast algorithms, there are various approaches in the literature, that include using dynamic programming~\cite{Uno12}, generating functions~\cite{brams1976power,bilbao2000,MannShapley62}, binary decision diagrams~\cite{Bolus2011BDD}, Monte-Carlo methods~\cite{BachrachMRPRS10,RM-2651,UshiodaTM22}, and so on.
The current best time complexity~\cite{kurz2016computing,Uno12} is $O(nq)$ for computing the Banzhaf index of all $n$ players with a given quota $q$, and $O(n^2q)$ for the Shapley-Shubik index of all $n$ players.
%See Section~\ref{sec:relatedwork} for further related work.
It should be remarked that we here assume, as previous papers,  that the arithmetic operations of $n$-bit numbers can be performed in constant time. 
See also a remark in Section~\ref{sec:conclusion}.

The main contribution of this paper is to design efficient algorithms for calculating the Banzhaf and Shapley--Shubik indices for all players, respectively, using techniques in theory of formal power series.
Our proposed algorithms run in $\rmO(n+q\log (q))$ time for the Banzhaf index, and $\rmO(nq\log (q))$ time for the Shapley--Shubik index.
Our algorithms are faster than existing algorithms when $q=2^{o(n)}$.

In our algorithms, we first represent power indices with generating functions~(i.e., polynomials).
%Generating functions are polynomial functions with the information of the power indices?
In the previous work such as~\cite{bilbao2000,brams1976power,MannShapley62}, generating functions are computed by dynamic programming.
In this paper, we regard a generating function as a formal power series, which is an infinite sum of monomials, and exploit the fast Fourier transform~(FFT) for efficient computation.
We remark that similar technique was recently used to solve the subset sum problem~\cite{JinW19}.

\subsection{Related work}\label{sec:relatedwork}

We here describe only algorithmic aspects of power indices in weighted majority games.
See also a survey~\cite{Matsui}.
Practical applications of power indices may be found in, e.g., ~\cite{ALGABA20071752,KOCZY2012152,kurz2016computing}.
%For the complexity of computing the power indices, see a survey.
%See \cite{} for other aspects.

As mentioned before, it is NP-hard to calculate the Banzhaf and Shapley--Shubik indices.
In fact, Deng and Papadimitriou~\cite{DengPapadimitriou94} showed the problem of computing the Shapley--Shubik index is \#P-complete.
Prasad and Kelly~\cite{PrasadKelly} proved that computing the Banzhaf index is \#P-complete, and that a number of decision problems around the Banzhaf and the Shapley–Shubik indices belong to the class of NP-complete problems.
See also~\cite{MatsuiMatsui01} for other hardness results.
It is known that even approximating the Shapley--Shubik index within a constant
factor is intractable, unless $\text{P} = \text{NP}$~\cite{ElkindGGW07}.

%Bolus~\cite{Bolus2011BDD} proposed a BDD-based algorithm running in expected $O(nq)$ time.

Cantor~(as mentioned by~\cite{MannShapley62}) represented the Shapley–Shubik index with a generating function~(a polynomial), which was used for calculation~\cite{Lucas1983,MannShapley62}.
%Mann and Shapley~\cite{MannShapley62} proposed a dynamic programming algorithm based on the generating function.
See Brams and Affuso~\cite{brams1976power}
for the Banzhaf index.
%adapted the algorithm to calculate the Banzhaf index. 
Bilbao et al.~\cite{bilbao2000} proposed an algorithm based on generating functions, running in $\rmO(nq)$ time to compute the Banzhaf index of one fixed player.

%The Banzhaf index is counting the number of subsets $S$ that satisfy $w(S)\geq q$ and $w(S\setminus \{p\})<q$, which has similarly to counting the number of knapsack solutions.
Matsui and Matsui~\cite{Matsui} designed dynamic programming algorithms for the Banzhaf and Shapley--Shubik indices, respectively, which is a similar approach to the one for counting the number of knapsack solutions.
Uno~\cite{Uno12} later improved their algorithms, which run in $\rmO(nq)$ and $\rmO(n^2q)$ time, respectively, for calculating the Banzhaf and Shapley–Shubik indices of all players.
See also~\cite{kurz2016computing} for slight improvement with other parameters.
%proposed dynamic-programming algorithms for calculating the Banzhaf and Shapley--Shubik indices running in $\rmO(nq)$ and $\rmO(n^2q)$ time, respectively.
Klinz and Woeginger~\cite{KlinzWoeginger2005} 
proposed an $\rmO(n^2 2^{n/2})$-time algorithm based on the enumeration technique.
There is another line of research to calculate the power indices approximately, see, e.g.,~\cite{BachrachMRPRS10,RM-2651,UshiodaTM22}.
%Monte-Carlo

\section{Preliminaries}\label{sec:Pre}

\subsection{Power Index}\label{sec:powerindex}

Let us formally introduce our setting.
In a weighted majority game, there are $n$ players.
Let $N = \{1, 2,\ldots, n\}$ be a set of $n$ players.
Each player $p$ in $N$ has a non-negative integer weight $w_p$.
We are also given a non-negative integer $q$, which is called a \textit{quota}.
Then each player votes for or against a decision.
A \textit{coalition} is a set $S\subseteq N$ of players who vote for the decision.
A coalition $S$ is called \textit{winning} if $w(S) \geq q$, where we define $w(S)=\sum_{p\in S}w_p$, and \textit{losing} otherwise.
%Throughout this paper, we assume that there exists at least one winning coalition in the given game, that is, $w(N)\geq q$.

%The Banzhaf index $\Bz_p$ of player $p$ is defined as follows.
%\begin{defn}
    For a player $p$, the \textit{Banzhaf index} $\Bz_p$ is defined as 
\[
%   \begin{equation}\label{Banzhaf}
        \mathrm{Bz}_p := \frac{\#\set{S\subseteq N \mid w(S)\geq q, \ w(S\setminus \{p\})<q}}{2^n}.
%        \mathrm{Bz}_p := \frac{\#\set{S\subseteq N \mid i\ \mathrm{is\ swing\ w.r.t.}\ S}}{2^n}.
\]
%\end{equation}
%\end{defn}
Intuitively, for a winning coalition $S$ with $p\in S$,
the player $p$ is considered to have a power in $S$ if $S\setminus \{p\}$ is a losing coalition.
The Banzhaf index  $\Bz_p$ is equal to the probability that the player $p$ belongs to a coalition with having a power under the assumption that every coalition occurs uniformly at random.

We next define the Shapley--Shubik index. 
Consider the situation where, at first, a coalition $S$ is the empty set, and players join the coalition $S$ one by one in an order.
Then, at the beginning of this situation, $S$ is a losing coalition, and $S$ becomes a winning one at some point.
Let $p$ be the player such that $S$ has changed from a losing coalition to a winning one when $p$ has been added to $S$.
We can naturally think that the player $p$ has a power.

More specifically, let $\Pi$ be the set of permutations with length $n$.
We note that $|\Pi|=n!$.
A permutation $\pi\in \Pi$ indicates an ordering of players to join a coalition, that is, players make coalitions $S_1(\pi), S_2(\pi), \dots, S_n(\pi)$ in the order, where $S_i(\pi) =\{\pi_1, \pi_2,\dots, \pi_i\}$ for $i=1,2,\dots, n$. 
Then there exists a unique player $\pi_k$ such that 
$w(S_{k-1}(\pi)) < q$ and $w(S_k(\pi)) \geq q$.
%$\sum_{j=1}^{k-1}w_{\pi_j} < q$ and $\sum_{j=1}^{k}w_{\pi_j} \geq q$.
%In this case, player $\pi_k$ is called a \textit{pivot with respect to $\pi$}.
%The Shapley--Shubik index $\SSS_p$ of player $p$ is defined as follows.
The \textit{Shapley--Shubik index} $\SSS_p$ of player $p$ is defined as
%\begin{equation}
\[
        \SSS_p \coloneq \frac{\#\{ \pi\in \Pi\mid \exists k \text{ s.t. }w(S_{k-1}(\pi)) < q, \ w(S_{k}(\pi))\geq q, \ p=\pi_k\} }{n!}.
        %\SSS_p \coloneq \frac{\#\{ \pi\in \Pi\mid \sum_{j=1}^{k-1}w_{\pi_j} < q, \sum_{j=1}^{k}w_{\pi_j} \geq q, p=\pi_k\} }{n!}.
        %        \SSS_p \coloneq \frac{\#\{ \pi\in \Pi\mid i\ \mathrm{is\ pivot\ w.r.t.}\ \pi\} }{n!}
\]
The Shapley--Shubik index is equal to the probability that the player $p$ has a power under the assumption that every permutation occurs uniformly at random.
%\end{equation}
%\end{defn}

\subsection{Formal Power Series}\label{sec:FPS}

    Let $R$ be a commutative ring where the zero element and the identity are denoted by $0$ and $1$, respectively. 
    A \textit{formal power series $f$ over $R$} is an infinite sum of the form
    \[
     f = \sum_{i=0}^{\infty} a_ix^i = a_0 + a_1 x + a_2 x^2+\cdots,
    \]
    where $a_i\in R$ for every non-negative integer $i$.
    We call $a_i$ the \textit{$i$-th coefficient of $f$}.
    %A formal power series is indentical with an infinite sequence $(f_0, f_1, \dots)$.
    The $i$-th coefficient of $f$ is denoted by $[x^i]f$.
    %For a formal power series $f = \sum_{i=0}^{\infty}f_ix^i$, we denote $f_i = [x^i]f$ for $i\in \mathbb{Z}\cup\{\infty\}$.

    We denote the set of all formal power series over $R$ by $R[[x]]$, that is, $R[[x]] = \left\{\sum_{i=0}^{\infty} a_ix^i \relmiddle| a_i\in R\right\}$.
    Then $R[[x]]$ forms a ring.
%    The ring consisting of all formal power series over $R$ is called formal power series ring and is denoted by $R[[x]]$. 
Specifically, the ring $R[[x]]$ defines the following operations for two formal power series $f=\sum_{i=0}^{\infty}a_ix^i$ and $g=\sum_{i=0}^{\infty}b_ix^i$: 
  \begin{enumerate}
    \item (Sum) 
    \[
    f+ g = \sum_{i=0}^{\infty}(a_i + b_i)x^i.
    \]
    \item (Product) 
    \[
    f\cdot g 
    %= \left(\sum_{i=0}^{\infty}a_ix^i\right) \cdot \left(\sum_{i=0}^{\infty}b_ix^i\right) 
    = \sum_{k = 0}^{\infty}\left(\sum_{i+j = k}a_i \cdot b_j\right)x^k.
    \]
  \end{enumerate}
%\end{defn}

%For a formal power series $f$, if there exists some integer $d$ such that $[x^i]f=0$ for all $i >d$, then $f$ is a \textit{polynomial}.
%The integer $d$ is called the \textit{degree} of the polynomial $f$.

The \textit{(multiplicative) inverse} of a formal power series $f$ is a formal power series $g$ such that $f\cdot g = 1$.
The inverse of $f$, which is known to be unique~(if exists), is denoted by $1/f$ or $f^{-1}$.
For example, the inverse of $1-x$ is equal to $1+x+x^2+\cdots = \sum_{i=0}^{\infty}x^i$, as $(1-x)(1+x+x^2+\cdots)=1$.
%For example, the inverse of $1-ax$~($a\in R$) is $1+ax+(ax)^2+\cdots = \sum_{i=0}^{\infty}(ax)^i$.
Note that $f$ may not necessarily have the inverse.
An element having the inverse is called a \textit{unit}.

We observe that a formal power series $f\in R[[x]]$ is a unit if and only if $[x^0]f\in R$ has the inverse in a ring $R$.
In particular, when $R$ is a field, $f\in R[[x]]$ is a unit if and only if $[x^0]f\in R$ is not the zero element $0$.

%We next summarize properties on a formal power series ring on $\mathbb{Q}$ or $\mathbb{Q}[[y]]$.

%Assume that $R=\mathbb{Q}$ or $R=\mathbb{Q}[[y]]$.
A polynomial is a formal power series with only finitely many non-zero terms.
It is well-known that the product of two polynomials over the real field $\mathbb{R}$ can be computed efficiently by the fast Fourier transform~(FFT).
See e.g.,~\cite{kleinberg2006algorithm}.

%When $R$ is the real field $\mathbb{R}$, the following lemma is shown by the fast Fourier transform~(FFT).
\begin{lemma}\label{lem:product}
    Suppose that we are given two polynomials $f = \sum_{i=0}^{d}a_{i}x^{i}$ and $g = \sum_{j = 0}^{d}b_{j}x^{j}$ over the real field $\mathbb{R}$.
%    \begin{align*}
%        f = \sum_{i=0}^{d}a_{i}x^{i} \quad\text{and}\quad
%        g = \sum_{j = 0}^{d'}b_{j}x^{j},
%    \end{align*}
    Then their product $f\cdot g$ can be computed in $\rmO(d\log (d))$ time.
    That is, we can compute, in $\rmO(d\log (d))$ time, all the $k$-th coefficients 
    \begin{equation*}
        [x^k] (f\cdot g) = \sum_{i+j=k}a_{i}\cdot b_{j}
    \end{equation*}
    for non-negative integers $k$ with $k\leq 2d$.
\end{lemma}

By the above lemma, we can efficiently compute the product of two formal power series in the following sense.
For two formal power series $f$ and $g$, if we are given the first $t$ coefficients $[x^i]f$ and $[x^i]g$ for $i=0,1,\dots, t-1$, then we can obtain the first $t$ coefficients $[x^i](f\cdot g)$ of $f\cdot g$ for $i=0,1,\dots, t-1$, in $\rmO(t \log(t))$ time.
Note that, since a formal power series is an infinite sum, we set a parameter $t$ indicating the number of coefficients we want to compute.

\section{Computing the Banzhaf index}\label{sec:Banzhaf}

In this section, we present a fast algorithm for computing the Banzhaf index in weighted majority games.
Let $N=\{1,2,\dots, n\}$ be a set of players where player $p$ has a non-negative integer weight $w_p$, and let $q$ be a quota.
We may assume that $1\leq w_p < q$ for any player $p$, and that there exists at least one winning coalition in the given game, that is, $w(N)\geq q$.

The main result of this section is the following.
%The above approach was applied by Brams and Affnso (1976) for computing the normalized Banzhaf index.

%\begin{thm}\label{mainbz}
\begin{theorem}\label{mainbz}
    For a weighted majority game with $n$ players and a quota $q$, we can compute the Banzhaf index for all players in $\rmO(n+q\log (q))$ time.
\end{theorem}

It is known in~\cite{brams1976power} that the Banzhaf index $\mathrm{Bz}_p$ for a player $p$ can be represented with a polynomial over the real field $\mathbb{R}$.
For a player $p$, we define $f_p \in \mathbb{R}[[x]]$ as
\begin{equation*}
    f_p = \prod_{\ell\neq p}(1+x^{w_\ell}).
\end{equation*}
Then we observe that, for every non-negative integer $i$, the $i$-th coefficient $[x^i]f_p$ is equal to the number of subsets $S\subseteq N\setminus\{p\}$ with $w(S)=i$.
%\textbf{generating function?}

\begin{lemma}[Brams and Affuso~\cite{brams1976power}]\label{lem:BramsAffuso}
    For a player $p$, it holds that 
    \begin{equation*}%\label{Bzeq1}
        2^n\cdot \Bz_p = \sum_{j = q-w_p}^{q-1} [x^j]f_p.
    \end{equation*}
\end{lemma}

In Lemma~\ref{lem:BzFPS} below, we shall represent the Banzhaf index of all players with one formal power series over the real field $\mathbb{R}$.
    Define  $\hat{f} \in \mathbb{R}[[x]]$ as
    \begin{equation*}%\label{fff}
        \hat{f} = \prod_{p=1}^{n}(1+x^{w_p}).
    \end{equation*}
    Then, for each player $p$, we have
\[
    f_p = \frac{\hat{f}}{1+x^{w_p}}.
    \]

%    By the definition of $g$, together with~\eqref{Bzeq1}, we obtain the following.

\begin{lemma}\label{lem:BzFPS}
For each player $p$,  it holds that
    \begin{align}\label{Bzg2}
        2^n\cdot \Bz_p &= [x^{q-1}]\frac{f_p}{1-x} - [x^{q-w_p-1}]\frac{f_p}{1-x}\notag\\
         &= [x^{q-1}]\frac{\hat{f}}{(1-x)(1+x^{w_p})} - [x^{q-w_p-1}]\frac{\hat{f}}{(1-x)(1+x^{w_p})}.
    \end{align}
\end{lemma}
\begin{proof}
Since it holds that
%\[
$
\frac{1}{1-x} = \sum_{i=0}^{\infty}x^i, 
$
%\]
we have 
    \begin{equation*}%\label{Cumsum}
        \frac{f_p}{1-x} = f_p\cdot \left(\sum_{i=0}^{\infty}x^i\right) = \sum_{i=0}^{\infty}\left(\sum_{j=0}^{i}a_{pj}\right)x^i,
    \end{equation*}
    where we denote $a_{pj}=[x^j]f_p$ for every non-negative integer $j$.
Therefore, it holds that
\[
[x^{q-1}]\frac{f_p}{1-x} - [x^{q-w_p-1}]\frac{f_p}{1-x}
= \sum_{j=0}^{q-1}a_{pj} - \sum_{j=0}^{q-w_p-1}a_{pj}
= \sum_{j=q-w_p}^{q-1}a_{pj},
\]
which is equal to $2^n\cdot \Bz_p$ by Lemma~\ref{lem:BramsAffuso}.
Thus the lemma holds.
\end{proof}

We compute the Banzhaf index for all players using the formal power series in~\eqref{Bzg2}.
The proposed algorithm is described as follows.
%We here define 
%    \begin{equation}\label{Bzg}
%        g = \frac{\hat{f}}{1-x}.
%    \end{equation}

\begin{description}
  \item[Step 1.] Compute the $i$-th coefficients of $\hat{f}$ for all non-negative integers $i$ with $i < q$.% in $\rmO(n+q\log (q))$ time.
  \item[Step 2.] Compute the $i$-th coefficients of $g= \frac{\hat{f}}{1-x}$ for all non-negative integers $i$ with $i< q$.% in $\rmO(q)$ time.
  \item[Step 3.] For each player $p$ with different weights, compute the Banzhaf index $\mathrm{Bz}_p$ by~\eqref{Bzg2}.
  %by Lemma~\ref{}.%, in $\rmO(n+q\log (q))$ time.
\end{description}

We next evaluate the computational complexity of each step in the proposed algorithm.

Jin and Wu~\cite{JinW19} showed as below that a polynomial $f$ in the form of $f=\prod_{i=1}^{n} (1+x^{s_i})$ can be calculated efficiently with the aid of Lemma~\ref{lem:product}.
Thus Step~1 can be performed in $\rmO(n + q\log (q))$ time.

\begin{theorem}[Jin and Wu \cite{JinW19}]\label{subsetsum}
    For $n$ positive integers $s_1, s_2, \ldots, s_n$, define a polynomial $f$ over $\mathbb{R}$ to be $f=\prod_{i=1}^{n} (1+x^{s_i})$.
    Then, for a positive integer $t$, the first $t$ coefficients of $f$ can be computed in $\rmO(n+t\log (t))$ time.
\end{theorem}

For Step~2, since 
\[
g=\hat{f}\cdot \left(\frac{1}{1-x}\right) = \sum_{i=0}^{\infty}\left(\sum_{j=0}^{i}\hat{a}_{j}\right)x^i,
\]
where $\hat{a}_{j} =[x^j]\hat{f}$ for every non-negative integer $j$, we have $[x^i]g = \sum_{j=0}^{i}\hat{a}_{j}$ for $i=0,1,\dots, q-1$.
Hence we can compute $[x^0]g, \dots, [x^{q-1}]g$ sequentially in $\rmO(q)$ time.

%\begin{lemma}\label{BzCumSum}
% The coefficients $[x^1]g, \dots, [x^q]g$ can be computed in $\rmO(q)$ time.
%\end{lemma}
%\begin{proof}
%    It can be done, since we have $[x^j]g = \sum_{k\leq j}[x^k]f = [x^{j-1}]g + [x^j]f$ by~\eqref{Cumsum}. 
%\end{proof}

The following lemma evaluates the time complexity of Step~3.

\begin{lemma}\label{mainlemmabz}
    Suppose that we are given the first $q$ coefficients of $g$, that is, we are given $[x^0]g, [x^1]g,\dots, [x^{q-1}]g$. 
    For a player $p$ with weight $w_p$, the Banzhaf index $\Bz_p$ can be computed in $\rmO\left(\frac{q}{w_p}\right)$ time.
\end{lemma}

\begin{proof}
We use~\eqref{Bzg2} to compute the Banzhaf index $\Bz_p$.
%    Since $\frac{1}{1+x^{w_p}} = \sum_{j=0}^{\infty}(-1)^jx^{jw_p}$, 
It holds that
\[
\frac{\hat{f}}{(1-x)(1+x^{w_p})}
=\frac{g}{1+x^{w_p}}
=
\left(\sum_{j=0}^{\infty}(-1)^jx^{jw_p}\right) \cdot  g.
\]
Then, for $k\leq q-1$, its $k$-th coefficient is equal to
\[
\sum_{i+jw_p=k} (-1)^j\cdot [x^i]g
=\sum_{j=0}^{\left\lfloor\frac{k}{w_p}\right\rfloor}(-1)^j\cdot [x^{k - jw_p}]g.
\]
    Since the right-hand side has at most $\frac{k}{w_p}+1=\rmO\left(\frac{q}{w_p}\right)$ terms, we can compute the $k$-th coefficient of $\frac{g}{1+x^{w_p}}$ in $\rmO \left(\frac{q}{w_p}\right)$ time, provided $[x^0]g, [x^1]g,\dots, [x^{q-1}]g$.
    Therefore, since the right-hand side of~\eqref{Bzg2} consists of two coefficients of $\frac{g}{1+x^{w_p}}$, $\Bz_p$ can be calculated in $\rmO\left(\frac{q}{w_p}\right)$ time.
%    \begin{align*}
%        [x^{q-1}]\frac{g}{1+x^{w_p}} &\equiv \sum_{j=0}^{\left\lfloor\frac{q - 1}{w_p}\right\rfloor}(-1)^j\cdot [x^{q - jw_p}]g \pmod{(x^q)} \\
 %       [x^{q - w_p - 1}]\frac{g}{1+x^{w_p}} &\equiv \sum_{j=0}^{\left\lfloor\frac{q - w_p - 1}{w_p}\right\rfloor}(-1)^j\cdot [x^{q - jw_p}]g \pmod{(x^q)}
 %   \end{align*}
%    Since both of the right-hand sides are the sum of $\rmO\left(\frac{q}{w_p}\right)$ terms, we can calculate each in $\rmO\left(\frac{q}{w_p}\right)$ time.
\end{proof}

We are now ready to prove Theorem \ref{mainbz}.

\begin{proof}[Proof of Theorem~\ref{mainbz}]
    As discussed above, Step~1 in the proposed algorithm takes $\rmO(n + q\log (q))$ time by Theorem~\ref{subsetsum}, and Step~2 takes $\rmO(q)$ time.
    It follows from Lemma~\ref{mainlemmabz} that $\Bz_p$ of a player $p$ can be computed in $\rmO\left(\frac{q}{w_p}\right)$ time.
    Since two players $p, p'$ with $w_p = w_{p'}$ satisfy $\Bz_p = \Bz_{p'}$, we can compute the Banzhaf index of all players by computing those of players with different weights.
    Since weights are integers between $1$ and $q-1$, the complexity for Step~3 can be bounded by   
    \begin{equation*}
        \rmO\left(\sum_{w=1}^{q-1} \frac{q}{w}\right) = \rmO(q\log (q)).
    \end{equation*}
    Therefore, the total complexity of the proposed algorithm is $\rmO(n + q\log (q))$.
\end{proof}

\section{Computing the Shapley--Shubik index}\label{sec:SS}

In this section, we show that the Shapley--Shubik index can be computed efficiently.

\begin{theorem}\label{mainss}
    For a weighted majority game with $n$ players and a quota $q$, we can compute the Shapley--Shubik index of all players in $\rmO(n+\tilde{n} q\log (q))$ time, where $\tilde{n}=\min\{n, q\}$.
\end{theorem}

To prove the theorem, we introduce a polynomial with two variables.
For a polynomial $f=\sum_{k=0}^{d} \sum_{j=0}^{d'} b_{jk} x^j y^k$ with two variables,
we denote $[y^kx^j]f_p=b_{jk}$ for non-negative integers $k, j$.
    We also denote $[x^j]f_p=\sum_{k=0}^{d} b_{jk} y^k$ and $[y^k]f_p=\sum_{j=0}^{d'} b_{jk} x^j$.

%we denote by $[y^kx^j]f_p$ the coefficient of the monomial with $y^kx^j$ in $f_p$.
%That is, if we can write $f_p$ as $f_p =\sum_{k=0}^{\infty} \sum_{j=0}^{\infty} b_{jk} x^j y^k$ by polynomial expansion,
%    then $[y^kx^j]f_p=b_{jk}$.

\subsection{Algorithm}\label{sec:SSalg}

%Similarly to the Banzhaf index, 
As shown by Cantor~(cf.~\cite{Lucas1983,MannShapley62}), the Shapley--Shubik index can be represented with a polynomial with two variables.
For a player $p$, we define a polynomial $f_p$ in two variables $x, y$ over the real field $\mathbb{R}$ as
    \begin{equation*}
        f_p = \prod_{\ell\neq p}(1+yx^{w_\ell}).
    \end{equation*}
    Then $[y^kx^j]f_p$  is equal to the number of subsets $S\subseteq N\setminus\{p\}$ such that $|S|=k$ and $w(S)=j$.

\begin{lemma}[Cantor~(cf.~\cite{Lucas1983,MannShapley62})]\label{Cantor}
    For a player $p$, it holds that
    \[    %\begin{equation}\label{SSeq1}
        n!\cdot \SSS_p = \sum_{k=1}^{n-1}\left(k!\cdot (n-1-k)!\sum_{j=q-w_p-1}^{q-1} [y^kx^j]f_p\right).
    %\end{equation}
    \]
%where $[y^kx^j]f_p$ denotes the coefficient of the monomial having $y^kx^j$ in $f_p$.
\end{lemma}

Since $w_p\geq 1$ for every player $p$, we see that, if $[y^kx^j]f_p\neq 0$, then $k\leq j$ holds.
Hence it suffices to consider $[y^kx^j]f_p$ for $k\leq j$ in the right-hand side, implying that     \begin{equation}\label{SSeq1}
        n!\cdot \SSS_p = \sum_{k=1}^{\tilde{n}-1}\left(k!\cdot (n-1-k)!\sum_{j=q-w_p-1}^{q-1} [y^kx^j]f_p\right),
    \end{equation}
%    \]
%we can compute $\SSS_p$ from $[y^kx^j]f_p$'s for $k < \min\{n, q\}$ and $j< q$.
where $\tilde{n}=\min\{n, q\}$.

Similarly to Lemma~\ref{lem:BzFPS} in Section~\ref{sec:Banzhaf}, 
the Shapley--Shubik index $\SSS_p$ for a player $p$ can be represented as follows.
Define 
    \[
        \hat{f} = \prod_{p=1}^n (1+yx^{w_p})
        \quad\text{and}\quad
        g = \frac{\hat{f}}{1-x}.
    \]
%    and $g = \frac{\hat{f}}{1-x}$.
    %, both of which are elements in $R[[x]]$.
%Then $\hat{f}$ is in $R[[x]]$, as the degree of $\hat{f}$ with respect to $y$ is at most $n$.

\begin{lemma}\label{lem:SS1}
   For each player $p$, it holds that
   \begin{equation}\label{SSg2}
        n!\cdot \SSS_p = \sum_{k=1}^{\tilde{n}-1}k! \cdot (n-1-k)!\left([y^kx^{q-1}]\frac{g}{1+yx^{w_p}}-[y^kx^{q-w_p-1}]\frac{g}{1+yx^{w_p}}\right).
 \end{equation}
\end{lemma}
%    The proof may be found in Appendix~\ref{sec:Appendix}.
\begin{proof}
It holds that
\[
\frac{g}{1+yx^{w_p}} = 
\frac{f_p}{1-x}
=
f_p\cdot \left(\sum_{i=0}^{\infty}x^i\right)
=
\sum_{i=0}^{\infty} \left(\sum_{j=0}^i a_{pj}\right) x^i,
\]
where we denote $a_{pj}=[x^j]f_p$ for every non-negative integer $j$.
We note that $a_{pj}$ is a polynomial in $y$.
The coefficient of $y^kx^i$ in $f_p/(1-x)$ is equal to 
\[
[y^k] \left(\sum_{j=0}^i a_{pj}\right)
= \sum_{j=0}^i [y^k] a_{pj}
= \sum_{j=0}^i [y^kx^j] f_{p}.
\]
Hence we have 
\begin{align*}
[y^kx^{q-1}]\frac{g}{1+yx^{w_p}}-[y^kx^{q-w_p-1}]\frac{g}{1+yx^{w_p}}
&=
\sum_{j=0}^{q-1} [y^k x^j] f_{p}
-
\sum_{j=0}^{q-w_p-1} [y^k x^j] f_{p}\\
&=
\sum_{j=q-w_p}^{q-1} [y^k x^j] f_{p},
\end{align*}
implying that the right-hand side of~\eqref{SSg2} is equal to $n!\cdot \SSS_p$ by~\eqref{SSeq1}. %Lemma~\ref{Cantor}.
%\[
%\sum_{k=1}^{n-1}\left(\sum_{j=0}^i a_{pj} x^i\right)\cdot k! \cdot (n-1-k)!,
%\]
%\qed
\end{proof}

Our proposed algorithm computes the right-hand side of~\eqref{SSg2} for a player $p$.
By~\eqref{SSg2}, it suffices to compute $[y^kx^j]\frac{g}{1+yx^{w_p}}$ for all $k< \tilde{n}$ and $j<q$.
To this end, we regard a two-variable rational function as a formal power series~(w.r.t.~$x$) over the formal power series ring~(w.r.t.~$y$).
%We note that, since we are interested in the coefficients of $y^k$ for $k\leq \tilde{n}$ by~\eqref{SSg2}, we  
Specifically, let $R=\mathbb{R}[[y]]/(y^{\tilde{n}})$, which is the quotient ring of $\mathbb{R}[[y]]$ by the ideal generated by $y^{\tilde{n}}$.
That is, an element in $R$ can be written as $\sum_{i=0}^{\tilde{n}-1}b_i y^i$ for $b_i \in \mathbb{R}$.
Note that $R$ is isomorphic to the quotient ring $\mathbb{R}[y]/(y^{\tilde{n}})$ of polynomials.
Consider the formal power series ring $R[[x]]$.
Then, a formal power series $f$ in $R[[x]]$ can be written as 
\[
f = \sum_{i=0}^{\infty}\left(\sum_{j=0}^{\tilde{n}-1}a_{ij}y^j\right)x^i,
\]
where $a_{ij}\in\mathbb{R}$.
The \textit{$i$-th coefficient $[x^i]f$} is equal to $\sum_{j=0}^{\tilde{n}-1}a_{ij}y^j$.
%, which is an element in $R=\mathbb{R}[[y]]/(y^{\tilde{n}})$.
%The \textit{two-variable formal power series ring} $R[[x, y]]$ is defined as $\left(R[[y]]\right)[[x]]$.
%On the other hand, $f$ can be written as $f = \sum_{i=0}^{\infty}\sum_{j=0}^{\tilde{n}-1}a_{ij}x^iy^j$.
We also denote $[y^jx^i] f = a_{ij}$ for every non-negative integers $i$ and $j < \tilde{n}$.

The proposed algorithm is presented as below.
%presented below, is similar to the one for the Banzhaf index, but uses~\eqref{SSg2}, instead of~\eqref{Bzg2}.
%Let $R=\mathbb{R}[[y]]/(y^{\tilde{n}})$.

\begin{description}
  \item[Step 1.] Compute the $i$-th coefficients of $\hat{f}\pmod{y^{\tilde{n}}}$ for all non-negative integers $i$ with $i < q$.
  %$[x^0]\hat{f}, \dots, [x^q]\hat{f}$. 
  %the $i$-th coefficients of $\hat{f}$ for all $i$ with $i\leq q$.% in $\rmO(n+q\log (q))$ time.
  \item[Step 2.] Compute the $i$-th coefficients of $g\pmod{y^{\tilde{n}}}$ for all non-negative integers $i$ with $i < q$.
%  $[x^0]g, \dots, [x^q]g$,  where $g = \frac{f}{1-x}$.
  %the $i$-th coefficient of $g = \frac{f}{1-x}$ for all $i$ with $i\leq q$.% in $\rmO(q)$ time.
  \item[Step 3.] For each player $p$ with different weights, compute the Shapley--Shubik index $\mathrm{SS}_p$ by~\eqref{SSg2}.
  %, in $\rmO(n+q\log (q))$ time.
\end{description}

To evaluate the time complexity of the proposed algorithm, we first show a variant of Theorem~\ref{subsetsum} as follows, where the proof is deferred to Section~\ref{sec:TwoVarFPS}.
The theorem implies that Step~1 can be executed in $\rmO (n+\tilde{n}q \log (\tilde{n}q))$ time.

\begin{theorem}\label{2_variable_exp-log}
    For $n$ positive integers $s_1, s_2, \ldots, s_n$, define a polynomial $f=\prod_{i=1}^{n} (1+yx^{s_i})$ over the real field $\mathbb{R}$.
    Then, for two positive integers $m, t$, we can compute all the coefficients $[y^kx^j]f$ for $j<t$ and $k<m$ in $\rmO(n+mt\log (mt))$ time.
\end{theorem}

Moreover, we have the following lemma for Step~3.

\begin{lemma}\label{mainlemmass}
    Suppose that we are given the first $q$ coefficients of $g$ modulo $y^{\tilde{n}}$, denoted by $a_0, a_1, \dots, a_{q-1}$.
    %denoted by $[x^0]g, [x^1]g,\dots, [x^{q-1}]g$. 
    For a player $p$ with weight $w_p$, the Shapley--Shubik index $\SSS_p$ can be computed in $\rmO\left(\frac{\tilde{n} q}{w_p}\right)$ time.
\end{lemma}

\begin{proof}
We use~\eqref{SSg2} to compute the Shapley--Shubik index $\SSS_p$.
It holds that
\[
\frac{g}{1+yx^{w_p}} = \left(\sum_{j=0}^{\infty}(-1)^j y^jx^{jw_p}\right)\cdot g.
\]
Then, for $k\leq q-1$, its $k$-th coefficient is
\[
[x^k]\frac{g}{1+yx^{w_p}}=\sum_{i+jw_p=k} (-1)^jy^j\cdot [x^i]g
=\sum_{j=0}^{\left\lfloor\frac{k}{w_p}\right\rfloor}(-1)^jy^j\cdot [x^{k - jw_p}]g.
\]
Hence, for $\ell=1,2,\dots, \tilde{n}$, we have 
\[
[y^\ell x^k] \frac{g}{1+yx^{w_p}}
=\sum_{j=0}^{\left\lfloor\frac{k}{w_p}\right\rfloor}(-1)^j\cdot [y^{\ell-j}x^{k - jw_p}]g.
\]
    Since the right-hand side has at most $\frac{k}{w_p}+1$ terms, we can compute it in $\rmO\left(\frac{q}{w_p}\right)$ time, provided $a_0, a_1, \dots, a_{q-1}$.
    Therefore, since the right-hand side of~\eqref{SSg2} can be calculated from $2(\tilde{n}-1)$ coefficients of $\frac{g}{1+yx^{w_p}}$, $\SSS_p$ can be obtained in $\rmO\left(\frac{\tilde{n}q}{w_p}\right)$ time.
\end{proof}

We now prove Theorem~\ref{mainss}.
\begin{proof}[Proof of Theorem~\ref{mainss}]
    As discussed above, Step~1 in the proposed algorithm takes $\rmO(n+\tilde{n}q\log (\tilde{n}q))$ time by Theorem~\ref{2_variable_exp-log}.
    Since $\tilde{n}=\min\{n, q\}\leq q$, the time complexity for Step~1 is $\rmO (n+\tilde{n}q\log (q))$.
    %if $n\leq q$, we have $\rmO(\tilde{n}q\log (\tilde{n}q)) = \rmO(nq\log (q))$, since $n=\tilde{n}\leq q$, and, if $n> q$, we obtain $\rmO(\tilde{n}q\log (\tilde{n}q)) =\rmO(n q\log (q))$, since $\tilde{n}=q$.
    %Thus, in each case, the time complexity for Step~1 is $\rmO (nq\log (q))$.
    Moreover, Step~2 takes $\rmO(\tilde{n}q)$ time, similarly to Section~\ref{sec:Banzhaf}. 
    
    It follows from Lemma~\ref{mainlemmabz} that the Shapley--Shubik index $\SSS_p$ of player $p$ can be computed in $\rmO\left(\frac{\tilde{n}q}{w_p}\right)$ time.
    Since two players $p, p'$ with $w_p = w_{p'}$ satisfy $\SSS_p = \SSS_{p'}$, we can compute the Shapley--Shubik index of all players by computing those of players with different weights. 
    The complexity for Step~3 can be bounded by   
    \begin{equation*}
        \rmO\left(\sum_{w=1}^{q-1} \frac{\tilde{n}q}{w}\right) = \rmO(\tilde{n}q\log (q)).
    \end{equation*}
    Therefore, the total complexity is $\rmO(n+\tilde{n} q\log (q))$, which completes the proof of the theorem.
    \end{proof}

\subsection{Two-variable Formal Power Series}
\label{sec:TwoVarFPS}

  This section is devoted to proving Theorem~\ref{2_variable_exp-log}.
  The proof extends the one of Theorem~\ref{subsetsum} by~\cite{JinW19} for the real field $\mathbb{R}$ to the case when $R=\mathbb{R}[[y]]/(y^{m})$, where $m$ is a non-negative integer.
%  Note that $\mathbb{R}[[y]]/(y^{n+1})$ is isomorphic to the ring $\{\sum_{i=0}^{n}a_i y^i \mid a_i\in \mathbb{R}\}$, whose element is a polynomial in $y\bmod (y^{n+1})$.
%  To this end, we introduce some notion over $R$, which is not necessarily $\mathbb{R}$.
%  Note that we are mainly interested in $R=\mathbb{R}[[y]]$.

  We first show that the product of two-variable polynomials over $\mathbb{R}$ can be computed by Lemma~\ref{lem:product}.
  %See Appendix~\ref{sec:Appendix} for the proof.

\begin{lemma}\label{lem:TwoVarProduct}
    Suppose that we are given two polynomials $f = \sum_{i=0}^{d_x}\sum_{j = 0}^{d_y}a_{ij}x^{i}y^{j}$ and $g = \sum_{i=0}^{d_x}\sum_{j = 0}^{d_y}b_{ij}x^{i}y^{j}$. %in $R[[x]]$ where $R=\mathbb{R}[[y]]$.
%    \begin{align*}
%        f = \sum_{i=0}^{d_x}\sum_{j = 0}^{d_y}a_{ij}x^{i}y^{j} \quad \text{and}\quad
%        g = \sum_{i=0}^{d_x}\sum_{j = 0}^{d_y}b_{ij}x^{i}y^{j}.
%    \end{align*}
    Then the product $f\cdot g$ can be computed in $\rmO(d_xd_y\log(d_xd_y))$ time.
    That is, we can compute, in $\rmO(d_xd_y\log(d_xd_y))$ time, all the coefficients
    \begin{equation*}
        [y^hx^k] (f\cdot g) = \sum_{i_1+i_2=k}\sum_{j_1+j_2=h}f_{i_1j_1}\cdot g_{i_2j_2}
    \end{equation*}
    for two non-negative integers $k, h$ satisfying $k\leq 2d_x$ and $h\leq 2d_y$.
\end{lemma}
\begin{proof}
%\begin{proof}[Proof of Lemma~\ref{lem:TwoVarProduct}]
    We transform $f$ and $g$ to one-variable polynomials by substituting $y=x^{2d_x+1}$.
    The resulting polynomials are denoted by $F$ and $G$, respectively, that is,
    \begin{align*}
        F = \sum_{i=0}^{d_x}\sum_{j= 0}^{d_y}a_{ij}x^{i+(2d_x+1)j} \text{\quad and\quad }
        G = \sum_{i=0}^{d_x}\sum_{j = 0}^{d_y}b_{ij}x^{i+(2d_x+1)j}.
    \end{align*}
    Then, since $F$ and $G$ are polynomials of degree $\rmO(d_xd_y)$, their product $F\cdot G$ can be computed in $\rmO(d_xd_y\log(d_xd_y))$ time by Lemma~\ref{lem:product}.
    Since it holds that 
    \begin{equation*}
        [x^{k+(2d_x+1)h}]F\cdot G = [y^hx^k]f\cdot g
    \end{equation*}
    for every integers $k, h$ with $k\leq 2d_x, h\leq 2d_y$,
    we can obtain the coefficients of the product $f\cdot g$.
%\qed
\end{proof}

%Assume that $R=\mathbb{Q}$ or $R=\mathbb{Q}[[y]]$.
%  Let $R=\mathbb{R}[[y]]/(y^n)$.
  For a formal power series $f\in R[[x]]$ such that $[x^0]f$ is nilpotent with respect to $R$,
  define
  \[
    \exp (f) = \sum_{i=0}^{\infty} \frac{f^i}{i!}.
  \]
  We recall that an element $a\in R$ is \textit{nilpotent} if there exists a positive integer $k$ such that $a^k=0$.
  Also, for a formal power series $f\in R[[x]]$ such that $[x^0](f-1)$ is nilpotent with respect to $R$, define 
  \[
    \log (f) = -\sum_{i=1}^{\infty} \frac{(1-f)^i}{i}.
  \]

%Assume that $R=\mathbb{Q}$ or 
%Let $R=\mathbb{R}$.
%  For a formal power series $f\in \bbR[[x]]$ such that $[x^0]f=0$,
%  define
%  \[
%    \exp f = \sum_{i=1}^{\infty} \frac{f^i}{i!}.
%  \]
%  Also, for a formal power series $f\in \bbR[[x]]$ such that $[x^0]f=1$, define 
%  \[
%    \log (f) = -\sum_{i=0}^{\infty} \frac{(1-f)^i}{i}.
%  \]

When $R=\mathbb{R}$, it is known that $1/f$, $\exp (f)$, and $\log (f)$ can be computed efficiently with the aid of Lemma~\ref{lem:product} and the Newton method as below.
%The time complexity is the same as computing the product.

\begin{lemma}[Brent \cite{Brent}]\label{fps-exp}
Let $R=\mathbb{R}$, and let $f$ be a formal power series over $R[[x]]$ such that the first $t$ coefficients $a_0, a_1,\dots, a_{t-1}$, where $a_i =[x^i]f$, are given.
%let $f=\sum_{i=0}^{\infty} a_i x^i$ be a formal power series over $R[[x]]$ such that $a_0, a_1,\dots, a_{m-1}$ are given.
Then the following statements hold.
\begin{enumerate}
\item  If $f$ is a unit, 
the first $t$ coefficients of $1/f$ can be computed in $\rmO(t\log (t))$ time.    
\item   
If $a_0=1$, then the first $t$ coefficients of $\log (f)$ can be computed in $\rmO(t\log (t))$ time.
\item  If $a_0=0$, then the first $t$ coefficients of $\exp(f)$ can be computed in $\rmO(t\log (t))$ time.
\end{enumerate}
\end{lemma}
   
Using the above lemma, together with the fact that $f=\exp(\log (f))$, Theorem~\ref{2_variable_exp-log} follows.
See~\cite{JinW19} for the details.
%and Appendix~\ref{sec:Appendix}.

In this section, we show analogous results to Lemma~\ref{fps-exp} when $R=\mathbb{R}[[y]]/(y^{m})$.
The proof idea is similar to that of Lemma~\ref{fps-exp}, but the difference is in computing the $0$-th coefficients of $\log(f)$ and $\exp(f)$, while it is trivial for the case when $R=\mathbb{R}$.   
%The proof of~(1) and~(3) is almost identical with that of Lemma~\ref{fps-exp}.
We here provide the proof for completeness.
%show only~(2) for computing the logarithm.
%The proof of~(1) and~(3) is provided in Appendix~\ref{sec:Appendix} for completeness.

\begin{lemma}\label{lem:twovarFPS}
Let $R=\mathbb{R}[[y]]/(y^{m})$ for a positive integer $m$, and let $f$ be a formal power series over $R[[x]]$ such that the first $t$ coefficients $a_0, a_1,\dots, a_{t-1}$, where $a_i =[x^i]f$ in $R$, are given.
%let $f=\sum_{i=0}^{\infty} a_i x^i$ be a formal power series over $R[[x]]$ such that $a_0, a_1,\dots, a_{m-1}$ in $R$ are given.
Then the following statements hold.
\begin{enumerate}
\item  If $f$ is a unit, 
the first $t$ coefficients of 
$1/f$ can be computed in $\rmO(mt\log (mt))$ time.    
\item  If $a_0-1$ is nilpotent, then the first $t$ coefficients of $\log (f)$ can be computed in $\rmO(mt\log (mt))$ time.
\item  If $a_0$ is nilpotent, then the first $t$ coefficients of $\exp(f)$ can be computed in $\rmO(mt\log (mt))$ time.
\end{enumerate}
\end{lemma}
%\begin{proof}[Proof of Lemma~\ref{lem:twovarFPS}~(1)]
\begin{proof}
\textbf{(1)}
Let $g^\ast=1/f$.
Since $f\cdot g^\ast=1$, we see that $[x^0]g^\ast = 1/a_0$.
Since $a_0$ is a unit in $R=\mathbb{R}[[y]]/(y^m)$, 
the inverse $1/a_0$ can be computed in $\rmO (m \log (m))$ time by Lemma~\ref{fps-exp}.
Let $g_1=1/a_0$. 
Then $g_1 \equiv g^\ast \bmod x^1$.

Suppose that, for some integer $k$, we have $g_k$ in $R[[x]]$ such that $[x^i]g_k = [x^i]g^\ast$ for any $i< k$ and $[x^i]g_k = 0$ for any $i> k$.
This means that $g_k \equiv g^\ast \pmod{x^{k}}$.
We also observe that, since $[x^i](f\cdot g_k) =[x^i](f\cdot g^\ast)$ for any $i< k$, we have $f\cdot g_k \equiv 1 \pmod{x^{k}}$.

We will show that, using $g_k$, we can obtain $g_{2k}$ in $R[[x]]$ such that $g_{2k} \equiv g^\ast \pmod{x^{2k}}$.
Specifically, given $g_k$ in $R[[x]]$, we define $g_{2k}$ as
\begin{equation}\label{eq:NewtonInverse}
g_{2k}\equiv 2g_{k} -fg^2_{k} \pmod{x^{2k}}.
\end{equation}

We show that $g_{2k}$ defined above satisfies $g_{2k} \equiv g^\ast \pmod{x^{2k}}$.
Since $f\cdot g_k \equiv 1 \pmod{x^{k}}$, \eqref{eq:NewtonInverse} implies that $g_{2k}\equiv 2g_k - g_k = g_k \pmod{x^{k}}$.
Hence it holds that $(g_{2k}-g_{k})^2\equiv 0\pmod{x^{2k}}$, which is equivalent to
\[
%g^2_{2k}-2g_{2k}g_{k} +g^2_{k}\equiv 0
%\quad\Longleftrightarrow\quad
g^2_{k}\equiv -g^2_{2k}+2g_{2k}g_{k} \pmod{x^{2k}}.
\]
Using this to~\eqref{eq:NewtonInverse}, we obtain
\begin{align}\label{eq:NewtonInverse2}
g_{2k}&\equiv 2g_{k} +f\cdot (g_{2k}^2-2g_{2k} g_k) \pmod{x^{2k}}\notag\\
\Longleftrightarrow\quad
(g_{2k}-2g_{k}) &\equiv f\cdot g_{2k} \cdot(g_{2k}-2g_{k}) \pmod{x^{2k}}.
\end{align}
Recall that $g^\ast$ has the inverse, and hence $[x^0]g^\ast$ has the inverse in $R$.
Since $[x^0](g_{2k}-2g_{k})=- [x^0]g_k = - [x^0]g^\ast$, we see that $(g_{2k}-2g_{k})$ also has the inverse.
Therefore,~\eqref{eq:NewtonInverse2} implies that $f\cdot g_{2k}\equiv 1\pmod{x^{2k}}$.
This means that $g_{2k} \equiv g^\ast \pmod{x^{2k}}$.

Thus we can compute $g_{2k}$ from $g_k$ by~\eqref{eq:NewtonInverse}, which takes $\rmO(mk\log (mk))$ time by Lemma~\ref{lem:TwoVarProduct}.
We repeat this procedure until the integer $k$ is larger than or equal to $t$.
In the end, we can obtain a desired one $g_t$ in $R[[x]]$ such that $[x^i]g_t = [x^i]g^\ast$ for any $i< t$.

Let $T(t)$ be the time complexity to compute $g_t$.
Then we have
\[
T(t)=T(t/2) + \rmO(mt\log (mt)) = \rmO(mt\log (mt)).
\]
Thus the lemma~(1) holds.

\noindent
\textbf{(2)}
It is known that $\log (f) = \int f'/f$, where $f'$ is the differential of $f$ and  $\int$ is an integral~(See e.g.,~\cite{ardila2015algebraic} for the details).
By definition, it holds that $f'=\sum_{i=1}^{\infty}i a_i x^{i-1}$.
Hence we know the first $t-1$ coefficients of $f'$, $[x^0]f', [x]f', \dots, [x^{t-2}]f'$, from $a_1,\dots, a_{t-1}$.
Also, the first $t$ coefficients of $1/f$ can be computed in $\rmO(mt\log (mt))$ time by the first statement of this lemma.

Let $f'\cdot (1/f)=\sum_{i=0}^\infty b_i x^i$.
Then it follows from Lemma~\ref{lem:TwoVarProduct} that the first $t-1$ coefficients $b_0,\dots, b_{t-2}$ can be computed in $\rmO(mt\log (mt))$ time.
We see that 
\[
\int \frac{f'}{f} = \sum_{i=1}^\infty \frac{b_{i-1}}{i} x^i+[x^0]\log (f).
\]
%the integral of $f'/f$ is equal to $\sum_{i=1}^\infty \frac{b_{i-1}}{i} x^i+[x^0]\log (f)$.
Therefore, since $\log (f) = \int f'/f$, we can obtain $[x^1]\log (f), \dots, [x^{t-1}]\log (f)$ from $b_0,\dots, b_{t-2}$.
%except for the constant $[x^0]\log (f)$.
The time complexity of this step is $O(mt \log (mt))$.

It remains to compute $[x^0]\log (f)$.
We observe that $[x^0]\log (f) = \log ([x^0]f)$.
In fact, it holds by definition that
\[
[x^0]\log(f)
=[x^0]\left(-\sum_{i=1}^\infty \frac{(1-f)^i}{i}\right)
=-\sum_{i=1}^\infty [x^0]\frac{(1-f)^i}{i}
=
-\sum_{i=1}^\infty \frac{(1-a_0)^i}{i},
%= \log (a_0).
\]
where the last equation follows from $[x^0](1-f)^i=(1-a_0)^i$.
%, the right-hand side is
%\[
%-\sum_{i=1}^\infty (1-a_0)^i/i) = \log (a_0).
%\]
%Now we see that $[x^0]\log (f) = \log ([x^0]f) = -\sum_{i=1}^\infty (1-a_0)^i/i)$.
The most right-hand side, which is $\log (a_0)$, is an element of $\mathbb{R}[[y]]/(y^m)$, and hence it can be computed in $\rmO(m \log (m))$ time by Lemma~\ref{fps-exp}.
Thus, $[x^0]\log (f)$ can be obtained in $\rmO(m \log (m))$ time.

Therefore, all the first $t$ coefficients of $\log (f)$ can be calculated in $\rmO(mt \log (mt))$ time in total, which completes the proof of the lemma~(2).

\noindent
\textbf{(3)}
Let $g^\ast =\exp (f)$.
Then $\log (g^\ast) = f$.
Also, define $G(h) = \log (h)- f$.

We first observe that $[x^0]\exp(f) = \exp([x^0]f)$, since we have by definition 
\[
[x^0]\exp(f)
=[x^0]\left(\sum_{i=0}^{\infty} \frac{f^i}{i!}\right)
=\sum_{i=0}^{\infty} \frac{[x^0]f^i}{i!}
=\sum_{i=0}^{\infty} \frac{a_0^i}{i!}
=\exp([x^0]f).
\]
 Since $[x^0]f$ is an element of $R=\mathbb{R}[[y]]/(y^m)$, $\exp([x^0]f)$ can be computed in $\rmO (m \log (m))$ time by Lemma~\ref{fps-exp} applied to $\mathbb{R}[[y]]$.
We set $g_1 = \exp([x^0]f)$.
%Since $[x^0]g^\ast = 1$, we set $g_0 = 1$\textbf{Is it true when $[x^0]g^\ast-1$ is nilpotent??}.

Assume that, for some integer $k$, we are given a formal power series $g_{k}$ in $R[[x]]$ that satisfies $g_k \equiv g^\ast \pmod{x^k}$.
In other words, $g_k$ satisfies that $G(g_k)\equiv 0\pmod{x^{k}}$.
%$[x^i]g_k = [x^i]g^\ast$ for any $i\leq k$.
%Below, we describe how to obtain $g_{2k}$ in $R[[x]]$ that satisfies $g_{2k} \equiv g^\ast \pmod{x^{2k}}$ using $g_K$.
Then we define $g_{2k}$ in $R[[x]]$ as 
\begin{equation}\label{eq:NewtonExp}
g_{2k} \equiv g_k - \frac{G(g_k)}{G'(g_k)} \pmod{x^{2k}},
\end{equation}
where $G'$ is the differential of $G$.

We shall show that $G(g_{2k})\equiv 0 \pmod{x^{2k}}$.
Since $G(g_k) \equiv 0 \pmod{x^{k}}$, it holds that $g_{2k}\equiv g_{k} \pmod{x^{k}}$ by~\eqref{eq:NewtonExp}.
Hence we have $(g_{2k}-g_k)^2\equiv 0 \pmod{x^{2k}}$.
Take the Taylor expansion of $G$ at $g_k$, that is, 
\begin{align*}
G(f_0) &= G(g_k) + G'(g_k) \left(f_0-g_k\right) + O\left((f_0-g_k)^2\right)
\end{align*}
for a formal power series $f_0$.
This implies that
\begin{align*}
G(g_{2k}) &= G(g_k) + G'(g_k) (g_{2k}-g_k) + O((g_{2k}-g_k)^2)\\
&\equiv G(g_k) + G'(g_k) (g_{2k}-g_k) \pmod{x^{2k}},
\end{align*}
since $(g_{2k}-g_k)^2\equiv 0 \pmod{x^{2k}}$.
Substituting~\eqref{eq:NewtonExp} for $g_{2k}$ in the equation, we have $G(g_{2k})\equiv 0 \pmod{x^{2k}}$.

Since $G'(g_k) = 1/g_k$, \eqref{eq:NewtonExp} is equal to
\[
g_{2k} \equiv g_k - g_k \left(\log (g_k) - f\right) = g_k \left(1-\log (g_k) + f\right) \pmod{x^{2k}}.
\]
%This is a desired $g_{2k}$.
Therefore, $g_{2k}$ can be computed from $g_k$ in $\rmO(mk\log (mk))$ time by Lemma~\ref{lem:TwoVarProduct}.

We repeat the above procedure until the integer $k$ is larger than or equal to $t$.
Then we can compute $g_t$ such that $G(g_t)\equiv 0 \pmod{x^{t}}$, which is a desired one, as $G(g_t)\equiv 0\pmod{x^{t}}$ means that $\log (g_t)\equiv f \pmod{x^{t}}$.
Let $T(t)$ be the time complexity to compute $g_t$.
Then we have
\[
T(t)=T(t/2) + \rmO(mt\log (mt)) = \rmO(mt\log (mt)).
\]
Thus the total complexity is $\rmO(mt\log (mt))$ time.
\end{proof}

We note that the above lemma does  not hold for a general commutative ring $R$, as it might happen that $[x^i]f$ is an infinite sum, and moreover, it is not clear how to compute the $0$-th coefficients.

We are now ready to prove Theorem~\ref{2_variable_exp-log}.

\begin{proof}[Proof of Theorem~\ref{2_variable_exp-log}]
Recall that $f = \prod_{i=1}^n (1+yx^{s_i})$ is a polynomial over the real field $\mathbb{R}$.
We consider $f$ as an element of the formal power series ring $R[[x]]$ where $R=\mathbb{R}[[y]]/(y^{m})$, and compute the first $t$ coefficients of $f\pmod{y^m}$ using the relationship $f=\exp (\log (f))$.
We may assume that $s_i\leq t$ for every $i$, as otherwise we may remove $(1+yx^{s_i})$ for $s_i>t$ from $f$. 

We first compute the first $t$ coefficients of $\log (f)$ as follows.
It holds that
    \begin{equation}\label{eq:thmTwoFPS}
        \log(f) = \log \prod_{i=1}^{n} (1+yx^{s_i}) = \sum_{i=1}^{n} \log (1+yx^{s_i}).
    \end{equation}
    Thus the first $t$ coefficients of $\log(f)$ can be computed from those of $\log (1+yx^{s_i})$ with different exponent $s_i$.
By the definition of the logarithm, we have
    \begin{equation*}
        \log(1+yx^{s_i}) = \sum_{j=1}^{\infty}(-1)^j\frac{(yx^{s_i})^j}{j}
        =\sum_{j=1}^{\infty}\frac{(-1)^jy^j}{j}x^{s_ij}. 
    \end{equation*}
Then $\log(1+yx^{s_i})$ has at most $\frac{t}{s_i}$ non-zero coefficients in the first $t$ coefficients, where each coefficient has only one monomial in $y$.
Hence the first $t$ coefficients of $\log(1+yx^{s_i})$ can be computed in $\rmO\left(\frac{t}{s_i}\right)$ time.
%Since each coefficient is a polynomial of degree at most $m-1$ in $y$, the first $t$ coefficients of $\log(1+yx^{s_i})$ can be obtained by computing at most $m\frac{t}{s_i}$ non-zero values.
Therefore, since $s_i\leq t$ for every $i$, it follows from~\eqref{eq:thmTwoFPS} that the first $t$ coefficients of $\log(f)$ can be computed in time
\[
\rmO\left(\sum_{s=1}^t \frac{t}{s}\right) = \rmO( t \log (t)).
\]

Finally, since $f=\exp(\log (f))$, we can compute the first $t$ coefficients of $f$ in $\rmO(mt\log (mt))$ time by Lemma~\ref{lem:twovarFPS}.
Thus the theorem holds, as reading the input takes $\rmO(n)$ time.
\end{proof}

\section{Conclusion}\label{sec:conclusion}

In this paper, we proposed fast pseudo-polynomial-time algorithms for computing power indices in weighted majority games.
Our algorithms calculate the Banzhaf index and the Shapley--Shubik index for all players in $\rmO(n+q\log (q))$ time and $\rmO(nq\log (q))$ time, respectively, where $n$ is the number of players and $q$ is a given quota. 
Our algorithms are faster than existing algorithms when $q=2^{o(n)}$.

We remark that our algorithms ignore the time complexity to deal with high-precision integers, following the literature such as~\cite{Matsui,Uno12}.
As the numerators of the Banzhaf and Shapley--Shubik indices may be $2^{\mathrm{\Omega}(n)}$, the time complexity in the standard computation model would increase by a factor of $n$.
On the other hand, we can avoid arithmetic operations of large integers by applying our algorithm $O(n)$ times as follows.
Let $\alpha_1,\dots, \alpha_\ell$ be prime numbers such that $q<\alpha_i\leq n$ for each $i$ and $2^n < \prod_i \alpha_i$, where $\ell \leq n$.
Then, for each $i$, we perform our algorithm over a prime field $\mathbb{F}_{\alpha_i}$ to obtain the output modulo $\alpha_i$.
It follows from the Chinese Remainder Theorem that we can uniquely recover the exact value of the output from these remainders.

In weighted majority games,
%esides the Banzhaf and Shapley--Shubik indices, 
other indices are proposed to measure the voting power of each player, such as the Deegan–Packel index and the Public Good index.
See e.g.,~\cite{AlonsoMeijide2012,kurz2016computing,StaudacherKSFKN21} and references therein.
It is known that most of the existing approaches for computing the Banzhaf and Shapley--Shubik indices can be adapted to these indices.
It would be interesting if our approach can be used to calculate these power indices.

\subsubsection{\ackname} 
This work was supported by
JSPS KAKENHI Grant Numbers 22H05001, 23K21646, and
JST ERATO Grant Number JPMJER2301, Japan.

\bibliographystyle{splncs04}
\bibliography{powerindex}

\begin{thebibliography}{10}
\providecommand{\url}[1]{\texttt{#1}}
\providecommand{\urlprefix}{URL }
\providecommand{\doi}[1]{https://doi.org/#1}

\bibitem{ALGABA20071752}
Algaba, E., Bilbao, J.M., Fern\'andez, J.R.: The distribution of power in the {E}uropean {C}onstitution. European Journal of Operational Research  \textbf{176}(3),  1752--1766 (2007). \doi{https://doi.org/10.1016/j.ejor.2005.12.002}, \url{https://www.sciencedirect.com/science/article/pii/S0377221705008805}

\bibitem{AlonsoMeijide2012}
Alonso-Meijide, J.M., Freixas, J., Molinero, X.: Computation of several power indices by generating functions. Applied Mathematics and Computation  \textbf{219}(8),  3395--3402 (2012). \doi{https://doi.org/10.1016/j.amc.2012.10.021}, \url{https://www.sciencedirect.com/science/article/pii/S0096300312010089}

\bibitem{ardila2015algebraic}
Ardila, F.: Algebraic and geometric methods in enumerative combinatorics. Handbook of Enumerative Combinatorics pp. 3--172 (2015)

\bibitem{BachrachMRPRS10}
Bachrach, Y., Markakis, E., Resnick, E., Procaccia, A.D., Rosenschein, J.S., Saberi, A.: Approximating power indices: theoretical and empirical analysis. Auton. Agents Multi Agent Syst.  \textbf{20}(2),  105--122 (2010). \doi{10.1007/S10458-009-9078-9}, \url{https://doi.org/10.1007/s10458-009-9078-9}

\bibitem{Banzhaf1965}
Banzhaf, J.F.: Weighted voting doesn't work: {A} mathematical analysis. Rutgers Law Review  \textbf{19}(2),  317--343 (1965)

\bibitem{bilbao2000}
Bilbao, J.M., Fernandez, J.R., Jim\'enez-Losada, A., L\'opez, J.: Generating functions for computing power indices efficiently. Top  \textbf{8}(2),  191--213 (2000)

\bibitem{Bolus2011BDD}
Bolus, S.: Power indices of simple games and vector-weighted majority games by means of binary decision diagrams. European Journal of Operational Research  \textbf{210}(2),  258--272 (2011). \doi{https://doi.org/10.1016/j.ejor.2010.09.020}, \url{https://www.sciencedirect.com/science/article/pii/S0377221710006181}

\bibitem{brams1976power}
Brams, S.J., Affuso, P.J.: Power and size: A new paradox. Theory and Decision  \textbf{7}(1),  29--56 (1976)

\bibitem{Brent}
Brent, R.P.: Multiple-precision zero-finding methods and the complexity of elementary function evaluation. In: Traub, J. (ed.) Analytic Computational Complexity, pp. 151--176. Academic Press (1976). \doi{https://doi.org/10.1016/B978-0-12-697560-4.50014-9}, \url{https://www.sciencedirect.com/science/article/pii/B9780126975604500149}

\bibitem{DengPapadimitriou94}
Deng, X., Papadimitriou, C.H.: On the complexity of cooperative solution concepts. Mathematics of Operations Research  \textbf{19}(2),  257--266 (1994). \doi{10.1287/MOOR.19.2.257}, \url{https://doi.org/10.1287/moor.19.2.257}

\bibitem{ElkindGGW07}
Elkind, E., Goldberg, L.A., Goldberg, P.W., Wooldridge, M.J.: Computational complexity of weighted threshold games. In: Proceedings of the Twenty-Second {AAAI} Conference on Artificial Intelligence. pp. 718--723. {AAAI} Press (2007), \url{http://www.aaai.org/Library/AAAI/2007/aaai07-114.php}

\bibitem{JinW19}
Jin, C., Wu, H.: A simple near-linear pseudopolynomial time randomized algorithm for subset sum. In: Fineman, J.T., Mitzenmacher, M. (eds.) 2nd Symposium on Simplicity in Algorithms, {SOSA} 2019. OASIcs, vol.~69, pp. 17:1--17:6. Schloss Dagstuhl - Leibniz-Zentrum f{\"{u}}r Informatik (2019). \doi{10.4230/OASICS.SOSA.2019.17}, \url{https://doi.org/10.4230/OASIcs.SOSA.2019.17}

\bibitem{KOCZY2012152}
^^c3^^81. K^^c3^^b3czy, L.: Beyond lisbon: Demographic trends and voting power in the {E}uropean {U}nion {C}ouncil of {M}inisters. Mathematical Social Sciences  \textbf{63}(2),  152--158 (2012). \doi{https://doi.org/10.1016/j.mathsocsci.2011.08.005}, \url{https://www.sciencedirect.com/science/article/pii/S0165489611000916}, around the Cambridge Compromise: Apportionment in Theory and Practice

\bibitem{kleinberg2006algorithm}
Kleinberg, J., Tardos, {\'{E}}.: Algorithm Design. Pearson Education India (2006)

\bibitem{KlinzWoeginger2005}
Klinz, B., Woeginger, G.J.: Faster algorithms for computing power indices in weighted voting games. Mathematical Social Sciences  \textbf{49}(1),  111--116 (2005). \doi{https://doi.org/10.1016/j.mathsocsci.2004.06.002}, \url{https://www.sciencedirect.com/science/article/pii/S016548960400068X}

\bibitem{kurz2016computing}
Kurz, S.: Computing the power distribution in the {IMF}. arXiv preprint arXiv:1603.01443  (2016)

\bibitem{Lucas1983}
Lucas, W.F.: Measuring Power in Weighted Voting Systems, pp. 183--238. Springer New York, New York, NY (1983). \doi{10.1007/978-1-4612-5430-0_9}, \url{https://doi.org/10.1007/978-1-4612-5430-0_9}

\bibitem{RM-2651}
Mann, I., Shapley, L.S.: Values of Large Games, IV: Evaluating the Electoral College by Montecarlo Techniques. RAND Corporation, Santa Monica, CA (1960)

\bibitem{MannShapley62}
Mann, I., Shapley, L.S.: Values of Large Games, VI: Evaluating the Electoral College Exactly. RAND Corporation, Santa Monica, CA (1962)

\bibitem{Matsui}
Matsui, T., Matsui, Y.: A survey of algorithms for calculating power indices of weighted majority games. Journal of the Operations Research Society of Japan  \textbf{43}(1),  71--86 (2000). \doi{10.15807/jorsj.43.71}

\bibitem{MatsuiMatsui01}
Matsui, Y., Matsui, T.: {NP}-completeness for calculating power indices of weighted majority games. Theor. Comput. Sci.  \textbf{263}(1-2),  305--310 (2001). \doi{10.1016/S0304-3975(00)00251-6}, \url{https://doi.org/10.1016/S0304-3975(00)00251-6}

\bibitem{PrasadKelly}
Prasad, K., Kelly, J.S.: {NP}-completeness of some problems concerning voting games. International Journal of Game Theory  \textbf{19}(1),  1^^e2^^80^^939 (Mar 1990). \doi{10.1007/BF01753703}, \url{https://doi.org/10.1007/BF01753703}

\bibitem{shapley:book1952}
Shapley, L.S.: A value for n-person games. In: Kuhn, H.W., Tucker, A.W. (eds.) Contributions to the Theory of Games II, pp. 307--317. Princeton University Press, Princeton (1953)

\bibitem{ShapleyShubik1954}
Shapley, L.S., Shubik, M.: A method for evaluating the distribution of power in a committee system. The American Political Science Review  \textbf{48}(3),  787--792 (1954), \url{http://www.jstor.org/stable/1951053}

\bibitem{StaudacherKSFKN21}
Staudacher, J., K{\'{o}}czy, L.{\'{A}}., Stach, I., Filipp, J., Kramer, M., Noffke, T., Olsson, L., Pichler, J., Singer, T.: Computing power indices for weighted voting games via dynamic programming. Operations Research and Decisions  \textbf{31}(2) (2021). \doi{10.37190/ORD210206}, \url{https://doi.org/10.37190/ord210206}

\bibitem{Uno12}
Uno, T.: Efficient computation of power indices for weighted majority games. In: Chao, K., Hsu, T., Lee, D. (eds.) 23rd International Symposium on Algorithms and Computation, {ISAAC} 2012. Lecture Notes in Computer Science, vol.~7676, pp. 679--689. Springer (2012). \doi{10.1007/978-3-642-35261-4\_70}, \url{https://doi.org/10.1007/978-3-642-35261-4\_70}

\bibitem{UshiodaTM22}
Ushioda, Y., Tanaka, M., Matsui, T.: Monte carlo methods for the shapley-shubik power index. Games  \textbf{13}(3), ~44 (2022). \doi{10.3390/G13030044}, \url{https://doi.org/10.3390/g13030044}

\end{thebibliography}

\clearpage
\appendix

\end{document}